\documentclass[journal, twocolumn,final]{IEEEtran}
\usepackage{amsfonts}
\usepackage{amssymb}
\usepackage{amsmath}
\usepackage{dsfont}
\usepackage{graphicx,subfigure}
\usepackage{enumitem}
\usepackage[normalem]{ulem}
\usepackage{bm}
\usepackage{cite}
\usepackage{bigstrut}
\usepackage{float}
\usepackage{balance}
\usepackage{lipsum}
\usepackage{psfrag}
\usepackage{xcolor}
\usepackage{xfrac}
\usepackage{hyperref}
\usepackage[ruled,vlined]{algorithm2e}
\usepackage{tikz}
\usepackage{pgfplots}
\pgfplotsset{compat=1.15}
\usepackage{mathrsfs}
\usepackage{diagbox}
\usetikzlibrary{arrows}
\usetikzlibrary{decorations.pathreplacing}
\usetikzlibrary{fadings}
\newtheorem{theorem}{Theorem}

\newtheorem{proposition}{Proposition}

\usepackage{mathtools}
\SetKwInput{KwInput}{Input}                
\SetKwInput{KwOutput}{Output}

\begin{document}
\IEEEoverridecommandlockouts
\newcommand\norm[1]{\left\lVert#1\right\rVert}
\newcolumntype{P}[1]{>{\centering\arraybackslash}p{#1}}

\title{Improving Reliability of Hybrid Bit-Semantic Communications for Cellular Networks}

\author{Nikos G. Evgenidis,~\IEEEmembership{Graduate Student Member,~IEEE}, Sotiris A. Tegos,~\IEEEmembership{Senior Member,~IEEE}, \\Panagiotis D. Diamantoulakis,~\IEEEmembership{Senior Member,~IEEE}, Ioannis Krikidis,~\IEEEmembership{Fellow,~IEEE}, \\
and George K. Karagiannidis,~\IEEEmembership{Fellow,~IEEE}
\thanks{N. G. Evgenidis, S. A. Tegos, P. D. Diamantoulakis, and G. K. Karagiannidis are with the Department of Electrical and Computer Engineering, Aristotle University of Thessaloniki, 54124 Thessaloniki, Greece (e-mails: nevgenid@ece.auth.gr, tegosoti@auth.gr, padiaman@auth.gr, geokarag@auth.gr).}
\thanks{I. Krikidis is with the Department of Electrical and Computer Engineering, University of Cyprus, 1678 Nicosia, Cyprus (e-mail: krikidis@ucy.ac.cy).}
\thanks{The work of N. G. Evgenidis was supported by the Hellenic Foundation for Research and Innovation (HFRI) under the 5th Call for HFRI PhD Fellowships (Fellowship Number: 21069).}

}
\maketitle
 
\begin{abstract}
Semantic communications (SemComs) have been considered as a promising solution to reduce the amount of transmitted information, thus paving the way for more energy-and spectrum-efficient wireless networks. Nevertheless, SemComs rely heavily on the utilization of deep neural networks (DNNs) at the transceivers, which limit the accuracy between the original and reconstructed data and are challenging to implement in practice due to increased architecture complexity. Thus, hybrid cellular networks that utilize both conventional bit communications (BitComs) and SemComs have been introduced to bridge the gap between required and existing infrastructure. To facilitate such networks, in this work, we investigate reliability by deriving closed-form expressions for the outage probability of the network. Additionally, we propose a generalized outage probability through which the cell radius that achieves a desired outage threshold for a specific range of users is calculated in closed form. Additionally, to consider the practical limitations caused by the specialized dedicated hardware and the increased memory and computational resources that are required to support SemCom, a semantic utilization metric is proposed. Based on this metric, we express the probability that a specific number of users select SemCom transmission and calculate the optimal cell radius for that number in closed form. Simulation results validate the derived analytical expressions and the characterized design properties of the cell radius found through the proposed metrics, providing useful insights.
\end{abstract}
\section{Introduction} \label{Sec:Intro}
Next-generation networks are envisioned to support many applications, most of which require large data transfer among end users and the backhaul network \cite{9955312,6G}. To satisfy such demands, the most common approach is to increase the capacity of the users, which has been extensively used in networks so far. However, this solution is not expected to be viable for next-generation networks since the increasing demand for bandwidth has led to the investigation of higher-spectrum technologies, which also face numerous challenges. In more detail, the severe path loss associated with these frequencies, combined with the more sophisticated hardware that should be utilized by the network users, makes the practical implementation of these technologies extremely difficult. With these in mind, a novel approach is required to improve the performance of future wireless networks. As such, a semantic-oriented paradigm has been proposed to enhance modern networks by leveraging the fact that many future applications are more concerned about the semantics of the data rather than the actual data themselves \cite{9955525}.

Recently, semantic communications (SemComs) have been considered as an attractive solution that can revolutionize communications by exploiting the underlying semantics of the transferred data. As identified by Shannon and Weaver in their seminal work \cite{6773024}, there are three levels of communication, level A, which addresses the technical problem of how accurately information is transmitted, level B, which addresses the semantic problem and refers to how informative a message is, and level C, which addresses the effectiveness problem. Based on this, SemComs can improve communication by integrating semantic information in the technical problem described by level A \cite{6773024}. This is possible by using common knowledge shared a priori between all parties in the form of knowledge bases to minimize the difference between the meaning of the transmitted messages and that of the recovered messages \cite{9955312,evgenidis2024multiple}. This approach can enhance reliability since an error at the bit level does not necessarily translate to an error at the semantic level, as explained in \cite{chafii2023scientific}. 

Although the semantic problem has been known for many decades, recent advancements in the field of deep neural networks (DNNs), such as natural language processing and image processing, have rekindled interest in it. In particular, by leveraging special DNN architectures that are mainly based on transformer modules, contextual relationships within texts and images can be identified and characterized, which in turn can be used to extract semantic information from the original data to be transmitted \cite{9955312}. Additionally, autoencoders have been shown to be a powerful tool for increasing the robustness of SemCom systems due to their ability to efficiently and effectively perform source and channel coding simultaneously \cite{9398576}. Therefore, artificial intelligence tools can extract semantic information from digital data and mitigate distortions caused by poor deep fading conditions.  

\subsection{Literature Review}
Recently, SemComs have attracted a lot of attention. In more details, basic ideas and challenges along with applications and performance metrics of semantic-aware systems were discussed in \cite{9679803,9955312,9955525}. In order to theoretically facilitate SemComs various approaches have been conducted from an information theoretic perspective \cite{https://doi.org/10.48550/arxiv.2201.01389, 6004632, 9814642}. Nevertheless, most works have focused on enabling a variety of applications by suggesting new DNN architectures to achieve accurate recovery of the original data after the reception and decoding of the transmitted data. Such applications are usually associated with speech and audio signals, image transfer and text contents. 

Regarding audio and speech signals, a novel DNN architecture, namely DeepSC-S, has been proposed \cite{9500590, 9450827} and its achievable signal-to-distortion was studied, while, in \cite{audioDiffuse}, a latent diffusion model that utilizes additional textual context is used to generate low-dimensional representations of the audio signals. On the other hand, semantic image transmission has been one of the most well-investigated case-studies, as it is crucial to enable many applications, including the internet of vehicles \cite{imageSeg1}. For this purpose, various architectures based on convolutional layers and visual transformers have been proposed. In more detail, in \cite{8683463}, DeepJSCC was introduced and shown to outperform conventional bit communications (BitComs) especially at low signal-to-noise ratio (SNR), while in \cite{10002903} a peak-to-average power ratio reduction technique was considered to allow DeepJSCC to be utilized in conjunction with waveforms such as orthogonal frequency division multiplexing (OFDM). In \cite{niu2023hybrid}, a hybrid system was considered, where part of an image is compressed and transmitted in the conventional way, while the rest is transmitted utilizing DeepJSCC and the receiver combines the two signals to achieve better perceptual similarity. To tackle the problem of varying fading channels, in \cite{DRL_InfoBottle}, a deep reinforcement learning approach combined with information bottleneck was considered to jointly minimize accuracy and transmission delay under such conditions. Additionally, in \cite{image2}, a convolutional based network implementing an additional adaptive channel attention module was considered to enhance peak SNR. In an attempt to better capture more complicated characteristics of images by considering the mean intersection over union, a novel semantic network that leverages perceptual-semantic loss was utilized to perform semantic segmentation over image data in \cite{imageIMoU}.

Except for audio and image contents, text transmission has also been explored in depth as it is common for many applications due to its frequent use in real life. In \cite{9398576}, a fundamental DNN architecture, called DeepSC, that leverages semantic encoding to reduce information and joint channel coding was proposed, while, in \cite{9763856}, the \textit{semantic rate} was introduced and a comparison between a pure SemCom and a pure BitCom system in terms of overall spectral efficiency was investigated. Aiming to extend the DeepSC architecture, in \cite{9885016}, a similar design was examined for serving two users, where the training loss function was modified to account for both users. In addition to this, a cooperative system with relay-forwarding was considered in \cite{cooperative_relay}, while a novel DNN architecture to reduce the amount of semantic information using semantic sentence forwarding was studied in \cite{semRelay}. In the same direction, a downlink and an uplink NOMA-based scheme were proposed in \cite{hybrid1} and \cite{SemNOMA}, respectively, to simultaneously serve two users with one utilizing SemCom and the other one BitCom, where the similarity was approximated though data regression to derive a tractable, yet accurate expression for the semantic rate. 

Nevertheless, regarding multiple users scenarios, the pre-existing infrastructure of BitCom networks along with potential limitations when SemComs are utilized have paved the way for hybrid BitCom-SemCom networks. Specifically, in \cite{meHybrid} the minimization of the sum delay of all users and the minimization of the maximum delay under practical similarity level thresholds were studied, where SemCom outperforms BitCom for mid-range SNR values. In the same direction, in \cite{hybrid2}, a machine learning framework was proposed to tackle the semantic throughput maximization and a fairness problem optimization for hybrid networks. Similarly, to enhance the weighted sum semantic rate, a hybrid network that utilizes a relay for semantic data transmission was explored in \cite{optTextAllocationRelay} and the optimal resource allocation for all users was found. 

\subsection{Motivation and Contribution}
Although multi-user hybrid BitCom-SemCom networks have recently received attention, gaps in the understanding of the performance of such systems remain. In \cite{urllc}, an ultra-reliable, low-latency communication system was considered, and the delay violation probability for delay-constrained users was analyzed. However, the outage probability of hybrid networks has not been investigated, even though it is crucial to facilitating reliable communication for all users. Additionally, most studies overlook the importance of similarity in achieving reliable data transfer, which is essential for applications requiring strict similarity levels. For instance, pharmaceutical content and instructions \cite{pharma} are extremely critical to maintain high similarity thresholds to be considered reliable, as the accuracy of the inferred message is vital to ensure health safety \cite{Shone2011-gs}. In more detail, while similarity is included in the semantic rate, no threshold requirement is considered, thus allowing semantic users to achieve great performance in terms of rate, but without considering the quality of the inferred messages. In addition to these, as emphasized in \cite{CommCompTradeoff, gai_sem_survey,semcom_gen_survey}, the utilization of DNN architectures, especially when many modules are included, can be challenging to perform in large scale networks, as resources like memory and computational capabilities are limited at each communication node. 

Motivated by these, in this work, we investigate the performance of a hybrid BitCom-SemCom cellular network in terms of reliability and resource management, while taking into account similarity constraints for each network user. The contributions of this work are listed below:

\begin{itemize}
    \item For the first time in the literature, we study the performance of a hybrid BitCom-SemCom network in terms of outage probability, where each user must satisfy a common similarity threshold in order to utilize SemCom transmission. With this in mind, the outage probability of a pure SemCom and a hybrid user is extracted in closed form. Moreover, we prove that the hybrid network always has better performance over the other two transmission benchmarks, BitCom and SemCom, where network outage is considered for all users and at least one user being in outage. 
    
    \item We introduce a generalized outage probability to evaluate the outage probability when the network is allowed to tolerate a certain range of users being in outage. For the most significant scenario that considers up to a specific number of users being in outage, the cell radius, below which a desired outage rate threshold can be achieved, is found and a closed-form expression is derived for the case of free-space path loss.

    \item Although SemCom can be preferable from both an outage and a semantic rate perspective, its utilization can be costly due to the specialized dedicated hardware that must be deployed, while it is also characterized by increased requirements in terms of memory and computational resources in order to support multiple users. With this in mind, we introduce semantic utilization as a metric to calculate the probability that a specific range of users utilize SemCom.
    Using this metric, we can design the cell radius such that the users selecting SemCom transmission are in a level that can be supported by the resources of the network infrastructure. Moreover, the optimal radius to maximize this probability is given in terms of network parameters. 

    \item Simulation results are performed that validate the derived theoretical expressions and provide design insights about the cell radius. As indicated by the results, the hybrid BitCom-SemCom network can achieve significant enhancement in communication reliability and spectral efficiency, while its flexibility with respect to the number of users that utilize SemCom transmission offers an attractive solution to the memory and computational resource limitations posed by practical networks.
\end{itemize}

\subsection{Structure}
The rest of this paper is organized as follows. In Section \ref{sec:SysMod}, the system model that includes BitCom and SemCom transmission methods is presented. In Section \ref{sec:cellularHybrid}, the hybrid user is defined through the derivation of the SNR cumulative distribution function (CDF). In Section \ref{sec:perfAnalysis}, the outage probability for each type of network is derived and a generalized outage probability as well as a semantic utilization metric are proposed and studied. In Section \ref{sec:numerical}, the simulation and theoretical results are presented, while, in Section \ref{sec:conc}, the conclusions of this work are discussed.

\section{System Model}\label{sec:SysMod}
Without loss of generality, we assume that $L$ users equipped with a single antenna are serviced by a base station (BS), which aims to transmit a task $S$ consisting of $P$ sentences to each user. The BS is assumed to be located in the center of a cell and equipped with a single antenna as well as additional specialized hardware that can support SemCom transmission except for the conventional BitCom transmission. With this in mind, the path-loss of the $l$-th user is given as $p_l(r_l) = \left( \frac{\lambda}{4 \pi} \right)^2 {r_{l}}^{-a} $, where $\lambda$ is the wavelength corresponding to the transmitting frequency $f_c$ of the BS, $r_l$ is the distance between the $l$-th user and the BS and $a$ is the exponent loss. Without loss of generality, the small scale channel fading of all users is assumed to follow the same Rayleigh distribution $h_l \sim CN\left(0, 1 \right)$, $\forall l \in \{ 1, \cdots , L \}$. In the investigated hybrid BitCom-SemCom network, the BS can select to communicate with each user by utilizing either SemCom or conventional BitCom, as illustrated in Fig. \ref{fig:sysmod}. To simultaneously serve all users, frequency domain multiple access (FDMA) is utilized by the BS, which transmits with a constant power $\mathcal{P}$ for all users over a total available bandwidth $W$ that is split equally among them. As such, the received SNR of the $l$-th user is equal to $g_l = c_L|h_l|^2r_l^{-a}$, where $c_L = L\mathcal{P} / (N_0 W) (\lambda / 4\pi)^2$ and $N_0$ is the noise spectral density. 

\subsection{SemCom Transmission}
In order to enable SemCom transmission, the proposed system utilizes the DeepSC network, which maps each sentence within $S$ to a sequence of real numbers via a semantic encoder and passes them through a channel encoder to mitigate the effect of channel fading. By considering the $j$-th sentence, denoted as $S_j$, the aforementioned procedure can be described as creating a semantically equivalent sentence $S_j'$, which is the output of the semantic encoder network, and then encoding $S'_j$ through the channel encoder network into a vector $\mathbf{x}_{j} = [x_{1}, \cdots, x_{kO_{j}}]$, where $O_{j}$ denotes the number of words in $S_j$ and $k$ is the number of outputs of the DNN for each word. By leveraging continuous-amplitude signals, the vector elements of $\mathbf{x}_{j}$ are then transmitted using discrete-time analog transmission (DTAT) \cite{10002903} and a channel decoder and semantic decoder are used on the receiver side to reconstruct the original data. 

\begin{figure}
    \centering
    \includegraphics[width=0.9\linewidth]{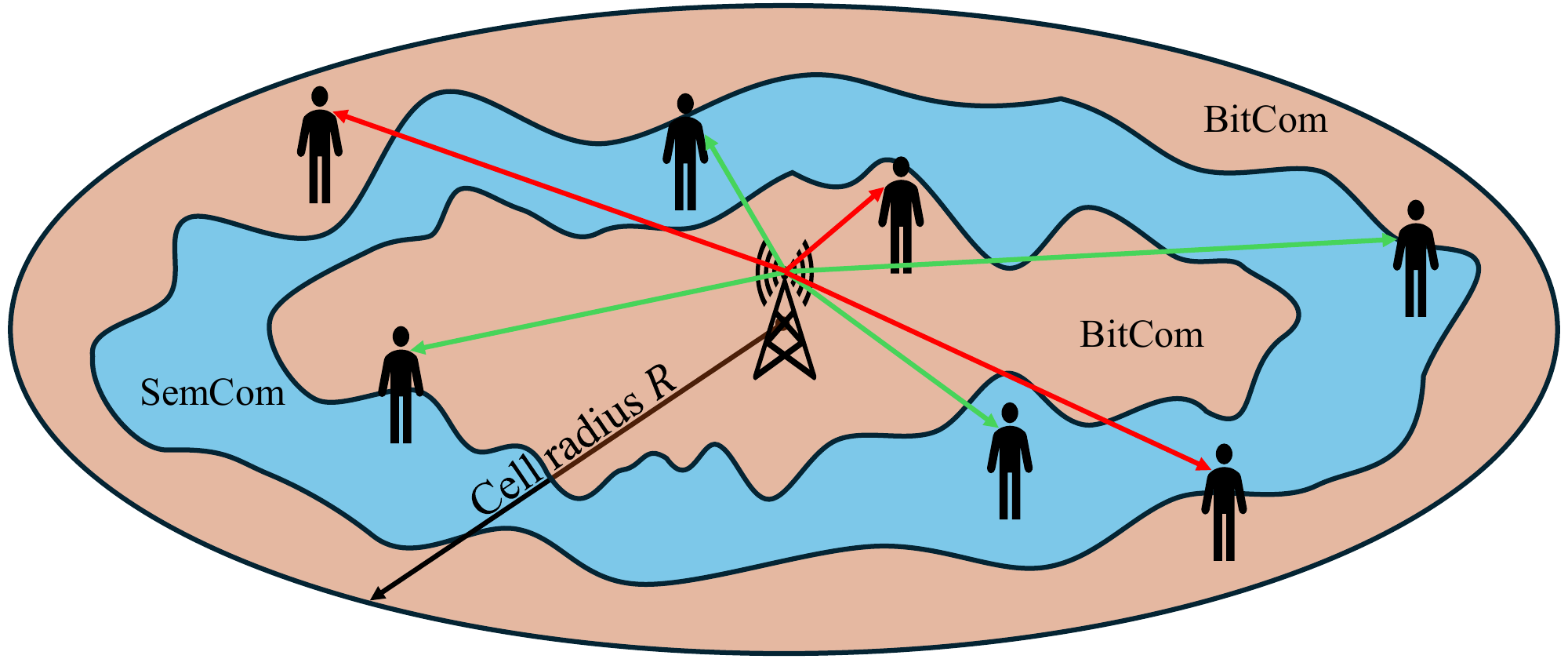}
    \caption{Model of the hybrid BitCom-SemCom network cell of multiple users. The users inside the blue region have a received SNR for which SemCom utilization is used, as illustrated by the green arrows. Similarly, the users in the brown areas have a received SNR such that BitCom transmission is utilized by the BS, as indicated by the red arrows.}
    \label{fig:sysmod}
\end{figure}

In order to measure the accuracy of an original sentence in $S_P$ and its reconstructed form, we use semantic cosine similarity as a performance metric, which is defined for the $j$-th sentence as
\begin{equation}\label{eq:SimDefine}
    M_{j} = \frac{\mathbf{B}(S_{j}) \mathbf{B}(S'_{j})^{T} }{ \|\mathbf{B}(S_{j}) \| \| \mathbf{B}(S'_{j})^{T} \|},
\end{equation}
where $(\cdot)^{T}$ denotes the transpose operator of a vector, and $\mathbf{B}(\cdot)$ denotes the bidirectional encoder representations from transformers (BERT) for each sentence $S_{j}$. Note that BERT transforms the original sentence $S_j$ into a vectorized representation of specific dimensions, which is then utilized to evaluate the cosine similarity between the vectorized form of the original sentence and its reconstructed counterpart. Semantic similarity as introduced in \eqref{eq:SimDefine} aims to evaluate the similarity of sentences across their entire length, thus allowing better estimation of semantic correlations between sentences compared to other metrics like bilingual evaluation understudy (BLEU) score which depends on $m$-grams, limiting the range of the semantic correlations that can be found by this performance metric. Considering the amount of semantic information $I_j$ that is conveyed by sentence $S_j$ of length $D_j$, we can introduce the expected values of semantic information per sentence $I = \sum_{j=1}^{P} I_j \mathbb{P}(S_j)$ and number of words per sentence $D = \sum_{j=1}^{P} D_j \mathbb{P}(S_j)$, respectively, where $\mathbb{P}(\cdot)$ denotes probability.

As illustrated in \cite{9398576,9763856}, the utilization of DeepSC leads to semantic similarity that varies according to the number of DNN outputs per word, $k$, and the received SNR, $g_l$, for all users $l \in \{1, \cdots, L\}$. Since the curves extracted from DeepSC do not have a closed-form expression, to handle the semantic similarity for further analysis, in this work we adopt the tractable and highly accurate approximation introduced in \cite{hybrid1}. Therefore, the semantic similarity for the $l$-th user is expressed through a generalized logistic function as 
\begin{equation}\label{eq:sem_approx}
    M_{k,l} = M_{k,l}(g_l) \approx A_{k,1} + \frac{A_{k,2} - A_{k,1}}{1+ e^{-\left( C_{k,1} 10\log_{10}(g_l) + C_{k,2} \right)}},
\end{equation}
where $A_{k,1}, A_{k,2}, C_{k,1}, C_{k,2}$ are curve-fitting derived constants that are dependent on $k$. Then, based on the average information and the average length of the whole task $S$, we can use the semantic rate of the $l$-th user as a measure of the transmission rate which is given as
\begin{equation}\label{eq:semRate}
    R_{s,l} = \frac{W_lI}{kD}M_{k,l},
\end{equation}
where $M_{k,l}$ is the semantic similarity of the $l$-th user and $R_{s,l}$ is measured in \textit{semantic symbols per second (suts/s)}. 

\subsection{BitCom Transmission}
Instead of utilizing semantic transmission, a user can also utilize BitCom to achieve a capacity equal to 
\begin{equation}\label{eq:capacity}
    C_l^{\mathrm{max}} = W_l \log_2\left( 1+ g_l \right).
\end{equation}
Nevertheless, the capacity is a theoretical limit and only close approximations to it can be achieved through specific coding schemes, such as low-density parity-check (LDPC) codes \cite{ldpc}. On the other hand, in the case of uncoded transmission, a rate gap between the capacity and the achievable maximum data rate exists due to the bit error rate (BER). Without loss of generality, we consider an uncoded $M$-QAM constellation, whose BER has been shown in \cite{634685,1177182} to be upper bounded as
\begin{equation} \label{eq:BERbound1}
    \mathrm{BER} \leq \frac{1}{5}\exp{\left( - 1.5 g_l \frac{1}{M-1} \right)},
\end{equation}
where $M$ is the modulation order. 
As a result, the maximum achievable data rate of an uncoded $M$-QAM scheme, which satisfies a required BER threshold is given by \cite{1177182}
\begin{equation}\label{eq:capacityGamma}
    C_{l} = W_l\log_2\left( 1+\frac{g_l}{\Gamma} \right),
\end{equation}
where $\Gamma = \sfrac{-\ln({\mathrm{5BER}})}{1.5}$ and is such that $\Gamma \geq 1 $. Thus, \eqref{eq:capacity} and \eqref{eq:capacityGamma} can be considered as the best and worst achievable data rates, respectively, with the latter reducing to the former when $\Gamma = 1$. To fairly compare the semantic and bit transmission, we use the equivalent expression of the BitCom rate \cite{9763856}, which is given as
\begin{equation}\label{eq:bitRate}
    R_{b,l} = C_{l}\frac{I}{\mu D} M_{k,l},
\end{equation}
where $\mu$ is a factor representing the average number of bits per word, thus $R_{b,l}$ is also measured in \textit{suts/s} for the $l$-th user. 

Since in BitCom accurate data transmission without errors is achieved, a perfect reconstruction of the transmitted data is possible at the receiver side of each user leading to $M_{k,l} = 1$. On the other hand, in SemCom, errors can be allowed at the receiver side of each user as long as a desired similarity threshold $M^{\mathrm{th}}$ between the transmitted and reconstructed data can be achieved. Therefore, the $l$-th user can utilize semantic transmission as long as $M^{\mathrm{th}} \leq M_{k,l} < 1$. Note that in the discussed system the BS selects a transmission mode for each user once at the beginning of the task and does not change until the transmission of the whole task is finished. 

\section{Hybrid BitCom-SemCom Cellular Network}\label{sec:cellularHybrid}
In the considered system model, all users are uniformly distributed within a cell of radius $R$. Taking into account that all users are subject to similar fading conditions described by the same distribution, hereinafter, the index is dropped from the received SNR, the distance from the BS, the similarity and the achievable BitCom and SemCom rate of each user for the sake of notation. With this in mind, the aforementioned will be denoted by $g$, $r$, $M_{k}$, $R_{b}$ and $R_{s}$, respectively. Therefore, each user follows the uniform circular distribution, i.e., $f_{r}(t) = \frac{2t}{R^2}$, $0 \leq t \leq R$, and the CDF of the received SNR of a user is given as
\begin{equation}\label{eq:CDFcombined}
    F_{g}(y) = \mathbb{E}_{r}\left[ F_{g | r}(y | r)\right] = \int_{0}^{R}\mathbb{P}(g \leq y | r = t)f_r(t) dt.
\end{equation}
Taking into consideration the fact that 
\begin{equation}\label{eq:probH}
    \mathbb{P}(g \leq y | r = t) = \mathbb{P}\left(|h|^2 \leq  \frac{yt^{a}}{c_L}\right) = 1 - e^{-\frac{yt^{a}}{c}},
\end{equation}
and by combining \eqref{eq:CDFcombined} with \eqref{eq:probH}, the CDF is calculated in closed form through the integral in \cite[eq. 3.381]{integralBook} as
\begin{align}\label{eq:CDF}
    F_{g}(y) &= 1 - \frac{2}{aR^2}\left( \frac{c_L}{y} \right)^{\frac{2}{a}}\gamma\left( \frac{2}{a},\frac{y}{c_L}R^a\right) \nonumber \\
    &= 1 - {}_1\!F_1\left(\!\frac{2}{a};1\!+\!\frac{2}{a};-\frac{y}{c_L}R^a\right),
\end{align}
where $\gamma(\cdot,\cdot)$ denotes the lower gamma incomplete function and ${}_1\!F_1(\cdot ; \cdot ; \cdot)$ is the confluent hypergeometric function of the first kind.

Due to the similarity threshold, imposed as a quality of service (QoS) constraint on a user, and the fact that BitCom outperforms SemCom in high received SNR values, it is obvious that a user should prefer to utilize semantic transmission when some bounds in received SNR are satisfied. These bounds can be equivalently expressed in terms of the received SNR of each user, $g$, thus valuable insights on the topology of the users in a cell to achieve a desired semantic utilization can be obtained. Considering a common similarity threshold $M^{\mathrm{th}}$, the minimum received SNR to achieve this QoS arises from $M_{k} = M^{\mathrm{th}}$ which has a closed-form solution in terms of received SNR given by
\begin{equation}\label{eq:SNRmin}
g_{\mathrm{min}} = 10^{- \frac{\ln\left( \frac{A_{k,2} - M^{\mathrm{th}}}{M^{\mathrm{th}} - A_{k,1}} \right) + C_{k,2}}{10C_{k,1}}},
\end{equation}
where the argument of the logarithm is always positive since $M^{\mathrm{th}} \in (A_{k,1}, A_{k,2})$. Accordingly, a maximum received SNR can be found by solving the equation $R_{s}(g) = R_{b}(g)$, which leads to 
\begin{equation}\label{eq:SNRmax}
g_{\mathrm{max}} = \{g | R_{s}(g) = R_{b}(g) \},
\end{equation}
beyond which BitCom outperforms its SemCom counterpart. Note that from \eqref{eq:sem_approx} it is straightforward to obtain $\lim_{g \to 0} M_{k} = A_{k,1}$ and $\lim_{g \to +\infty} M_{k} = A_{k,2}$, which essentially results in the semantic rate to be bounded as $R_{s} \in (A_{k,1},A_{k,2})$. On the other hand, $\lim_{g \to 0} R_{b}(g) = 0$ and $\lim_{g \to +\infty} R_{b}(g) = +\infty$, thus a solution that results in $g_{\mathrm{max}}$ always exists in \eqref{eq:SNRmax}. With this in mind, \eqref{eq:SNRmax} has a solution such that whenever $g_{\mathrm{max}} \leq g_{\mathrm{min}}$, BitCom always outperforms SemCom and the studied hybrid systems is equivalently reduced to the BitCom scenario. Therefore, the following set determines the range of SNR values for semantic utilization:
\begin{equation}\label{eq:SNRbounds}
  \mathcal{G} =
      \left\{g | g_{\mathrm{min}} \leq g \leq g_{\mathrm{max}} \right\}.
\end{equation}
Based on these, a hybrid user is defined to be one that selects between BitCom and SemCom transmission, as highlighted in Fig. \ref{fig:sysmod}, in a way that maximizes the semantic rate. Consequently, the rate of a hybrid user is expressed as
\begin{equation}\label{eq:hybridRate}
    R_{h} = \begin{cases}
        R_{b}(g),& g \leq g_{\mathrm{min}}, \\
        R_{s}(g),& g_{\mathrm{min}} \leq g \leq g_{\mathrm{max}},  \\
        R_{b}(g),& g_{\mathrm{max}} \leq g, 
    \end{cases}
\end{equation}
which can be derived by leveraging \eqref{eq:SNRbounds}.  

\section{Performance Analysis}\label{sec:perfAnalysis}
In this section, the outage probability of the proposed hybrid network is studied along with that of pure BitCom and SemCom networks. With this in mind, we first investigate two different versions of network outage, namely when all users are in outage and when at least one user is in outage. Additionally, a generalized outage metric is examined for the proposed hybrid network and the cell radius to satisfy a desired outage rate threshold is estimated according to it. Moreover, semantic utilization is investigated to find the optimal cell radius that ensures an efficient and effective resource management from the network side.

For a user of the studied hybrid network to be in outage, there are two possible cases. Specifically, a user outage can occur either if BitCom is preferred and outage in BitCom occurs or if SemCom is preferred and outage in SemCom occurs. 

Assuming that a rate threshold $R_{\mathrm{out}} I/D$ \textit{suts/s} is set below which outage occurs in BitCom, $R_{b} \leq R_{\mathrm{out}}$ holds for the values of SNR described by the following set:
\begin{equation}\label{eq:bitOut}
    \mathcal{A}_{\mathrm{b}} = \left\{ g | g \leq g_{\mathrm{bit}} \right\}, 
\end{equation}
where $g_{\mathrm{bit}} = 2^{\mu R_{\mathrm{out}}} - 1$. Similarly, considering the SemCom case, outage occurs when $R_{s} \leq R_{\mathrm{out}}$. The latter holds for the values of SNR described by the following set:
\begin{align}\label{eq:semOut}
    \mathcal{A}_{\mathrm{s}} = \begin{cases}
        \left\{ \emptyset\right\},& kR_{\mathrm{out}} \leq A_{k,1}, \\
        \left\{g | g \leq g_{\mathrm{sem}} \right\},& A_{k,1} \leq kR_{\mathrm{out}} \leq A_{k,2},  \\
        \mathbb{R},& A_{k,2} \leq kR_{\mathrm{out}}, 
    \end{cases}
\end{align}
where $g_{\mathrm{sem}}$ is derived in similar fashion to \eqref{eq:SNRmin} and is given as
\begin{equation}\label{eq:SNRoutSem}
g_{\mathrm{sem}} = 10^{- \frac{\ln\left( \frac{A_{k,2} - kR_{\mathrm{out}}}{kR_{\mathrm{out}} - A_{k,1}} \right) + C_{k,2}}{10C_{k,1}}}.
\end{equation}

The outage probability of a user in the hybrid network is defined as 
\begin{equation}\label{eq:outProbHyb}
    \Pi_h = \mathbb{P}\left(g \in \left\{ \mathcal{A}_b \cap \mathcal{G}^c \right\} \right) + \mathbb{P}\left(g \in\left\{\mathcal{A}_s \cap \mathcal{G} \right\} \right),
\end{equation}
where $\mathcal{G}^c$ denotes the complementary set of $\mathcal{G}$ and the first and second terms express the hybrid user outage probability for BitCom and SemCom, respectively. A closed-form expression for the outage probability is provided in the following theorem.
\begin{theorem}\label{th:HybridPerf}
The outage probability of a user in the hybrid network can be expressed in closed form through \eqref{eq:hybCombSemOut} which is given at the top of the next page.
\end{theorem}

\begin{IEEEproof}
Based on \eqref{eq:SNRbounds}, \eqref{eq:bitOut} and \eqref{eq:semOut}, the outage probability can be rigorously calculated through \eqref{eq:singleBitOut} and \eqref{eq:singleSemOut} for each case, given at the top of the next page. Using \eqref{eq:outProbHyb}, \eqref{eq:hybCombSemOut} is derived, which completes the proof.
\end{IEEEproof}

\begin{figure*}[!t]
\begin{equation}\label{eq:hybCombSemOut}
\Pi_h \!=\! \begin{cases}  F_{g}\left(g_{\mathrm{bit}}\right),& \!\!\! g_{\mathrm{bit}} \leq g_{\mathrm{min}}, kR_{\mathrm{out}} \leq A_{k,1}, \\ 
F_{g}\left(g_{\mathrm{min}}\right),&\!\!\! g_{\mathrm{min}} \leq g_{\mathrm{bit}} \leq g_{\mathrm{max}}, kR_{\mathrm{out}} \leq A_{k,1}, \\
F_{g}\left(g_{\mathrm{bit}}\right),&\!\!\! g_{\mathrm{sem}} \leq g_{\mathrm{bit}} \leq g_{\mathrm{min}}, A_{k,1} \leq kR_{\mathrm{out}} \leq A_{k,2}, \\
F_{g}\left(g_{\mathrm{min}}\right),&\!\!\! g_{\mathrm{sem}} \leq g_{\mathrm{min}} \leq g_{\mathrm{bit}} \leq g_{\mathrm{max}}, A_{k,1} \leq kR_{\mathrm{out}} \leq A_{k,2}, \\
F_{g}\left(g_{\mathrm{sem}}\right),&\!\!\! g_{\mathrm{min}} \leq g_{\mathrm{sem}} \leq g_{\mathrm{bit}} \leq g_{\mathrm{max}}, A_{k,1} \leq kR_{\mathrm{out}} \leq A_{k,2},\\
F_{g}\left(g_{\mathrm{bit}}\right),&\!\!\! g_{\mathrm{max}} \leq g_{\mathrm{bit}} \leq g_{\mathrm{sem}}, A_{k,1} \leq kR_{\mathrm{out}} \leq A_{k,2},\\ 
F_{g}\left(g_{\mathrm{bit}}\right),&\!\!\! g_{\mathrm{max}} \leq g_{\mathrm{bit}}, A_{k,2} \leq kR_{\mathrm{out}}\\ 
    \end{cases}
\end{equation}

\hrulefill
\vspace{-1mm}
\begin{equation}\label{eq:singleBitOut}
\!\!\mathbb{P}\left( g \! \in \!\left\{ \mathcal{A}_b \cap \mathcal{G}^c \right\} \right) \!=\! \begin{cases}        \mathbb{P}\left(g \leq g_{\mathrm{bit}} \right) = F_{g}\left( g_{\mathrm{bit}}\right) ,& \!\!\! g_{\mathrm{bit}} \leq g_{\mathrm{min}}, \\
\mathbb{P}\left(g \leq g_{\mathrm{min}} \right) = F_{g}\left( g_{\mathrm{min}}\right) ,&\!\!\! g_{\mathrm{min}} \leq g_{\mathrm{bit}} \leq g_{\mathrm{max}}, \\
\mathbb{P}\left( \left\{g \leq g_{\mathrm{min}} \right\} \!\cup \!\left\{g_{\mathrm{max}} \leq g \leq g_{\mathrm{bit}}\right\} \right) = F_{g}\left(g_{\mathrm{min}}\right) + F_{g}\left(g_{\mathrm{bit}}\right) -  F_{g}\left(g_{\mathrm{max}}\right),&\!\!\! g_{\mathrm{max}} \leq g_{\mathrm{bit}} 
    \end{cases}
\end{equation}

\hrulefill
\vspace{-1mm}
\begin{equation}\label{eq:singleSemOut}
\!\!\mathbb{P}\left(g \! \in \!\left\{ \mathcal{A}_s \cap \mathcal{G} \right\}\right) \!=\! \begin{cases} 
\mathbb{P}\left( \emptyset \right) = 0,&\!\!\! kR_{\mathrm{out}} \leq A_{k,1}, \\
\mathbb{P}\left( \emptyset \right) = 0,&\!\!\! g_{\mathrm{sem}} \leq g_{\mathrm{min}}, A_{k,1} \leq kR_{\mathrm{out}} \leq A_{k,2}, \\
\mathbb{P}\left(g_{\mathrm{min}} \leq g \leq g_{\mathrm{sem}} \right) = F_{g}\left(g_{\mathrm{sem}}\right) - F_{g}\left(g_{\mathrm{min}}\right),&\!\!\! g_{\mathrm{min}} \leq g_{\mathrm{sem}} \leq g_{\mathrm{max}}, A_{k,1} \leq kR_{\mathrm{out}} \leq A_{k,2},\\
\mathbb{P}\left(g_{\mathrm{min}} \leq g \leq g_{\mathrm{max}} \right) = F_{g}\left(g_{\mathrm{max}}\right) - F_{g}\left(g_{\mathrm{min}}\right),&\!\!\! g_{\mathrm{max}} \leq g_{\mathrm{sem}}, A_{k,1} \leq kR_{\mathrm{out}} \leq A_{k,2},\\ 
\mathbb{P}\left(g_{\mathrm{min}} \leq g \leq g_{\mathrm{max}} \right) = F_{g}\left(g_{\mathrm{max}}\right) - F_{g}\left(g_{\mathrm{min}}\right),&\!\!\! A_{k,2} \leq kR_{\mathrm{out}}
    \end{cases}
\end{equation}
\end{figure*}

Except for the outage probability of a hybrid network, it is also interesting to study the outage probability of both BitCom and SemCom when they are independently utilized.
In the case of pure BitCom communication, the user outage probability is expressed as
\begin{equation}\label{eq:outProbBit}
    \Pi_b =  \mathbb{P}\left(g \in \mathcal{A}_b  \right) = F_{g}\left(g_{\mathrm{bit}}\right).
\end{equation}
Similarly, in the case of SemCom, the user outage probability $\Pi_{s}$ is given by \eqref{eq:outProbSem} at the top of the next page.

\begin{figure*}
\hrulefill
\vspace{-1mm}
\begin{equation}\label{eq:outProbSem}
\Pi_{s} = \begin{cases}        \mathbb{P}\left(g \leq g_{\mathrm{min}} \right) \!= \!F_{g}\left( g_{\mathrm{min}}\right) ,& \!\!\! kR_{\mathrm{out}} \leq A_{k,1}, \\
\mathbb{P}\left(g \leq g_{\mathrm{min}} \right) = F_{g}\left( g_{\mathrm{min}}\right) ,& \!\!\! g_{\mathrm{sem}} \leq g_{\mathrm{min}}, A_{k,1} \leq kR_{\mathrm{out}} \leq A_{k,2}, \\
\mathbb{P}\left( g \leq g_{\mathrm{sem}} \right) = F_{g}\left(g_{\mathrm{sem}}\right),& \!\!\! g_{\mathrm{min}} \leq g_{\mathrm{sem}}, A_{k,1} \leq kR_{\mathrm{out}} \leq A_{k,2}, \\ 
\mathbb{P}\left( g \in \mathbb{R} \right) = 1,& \!\!\! A_{k,2} \leq kR_{\mathrm{out}}
    \end{cases}
\end{equation}
\hrulefill
\end{figure*}

\subsection{Network Outage}
Let $L_o$ denote the number of users that are in outage in the cell. To evaluate the performance of the proposed system in network level, we are primarily interested in two cases. In more detail, we define network outage to occur when all users are in outage or when at least one user is in outage, which are characterized by $L_o=L$ and $L_o \neq 0$, respectively.

By taking this into account, the network outage probability for the $x$-type network, where $x \in \{h,b,s\}$ denotes whether the users operate in the hybrid, BitCom or SemCom network, respectively, is given as
\begin{align}\label{eq:xNetOut}
    \Pi_{x\mathrm{-net}} = 
    \begin{cases}
        \mathbb{P}_{x}\left( L_o = L\right) = \Pi^L_{x}, \\
        \mathbb{P}_{x}\left( L_o \neq 0\right) = 1 - \left( 1 - \Pi_x\right)^L. 
    \end{cases}
\end{align}
Note that $\mathbb{P}_{x}\left( L_o = L\right)$ denotes the probability $\mathbb{P}\left( L_o = L\right)$ for the $x$-type network.

Taking this into account, it is possible to prove that the proposed hybrid network always achieves outage probability lower than or equal to the BitCom and SemCom network. In more detail, the following theorem can be stated regarding the outage performance of the three networks.

\begin{theorem}\label{th:NetworkPerf}
    The proposed hybrid network outperforms in terms of outage performance the pure BitCom and SemCom networks, i.e., $\Pi_{h\mathrm{-net}} \leq \Pi_{b\mathrm{-net}}$ and $\Pi_{h\mathrm{-net}} \leq \Pi_{s\mathrm{-net}}$.
\end{theorem}

\begin{IEEEproof}
    Note that equivalent semantic rates exist for $g_{\mathrm{min}}$ and $g_{\mathrm{max}}$ such that $R_{\mathrm{min}} = R_{s}(g_{\mathrm{min}})$ and $R_{\mathrm{max}} = R_{s}(g_{\mathrm{max}}) = R_{b}(g_{\mathrm{max}})$, respectively.
    To compare the performance between all networks for both definitions of network outage, it suffices to compare a single user for each network case. We consider four cases and show that for all of them the hybrid network outperforms the BitCom and SemCom networks in terms of reliability.

    \begin{enumerate}[wide, labelwidth=!, labelindent=0pt]
    \item If $kR_{\mathrm{out}} \leq A_{k,1}$:
 In this case, as illustrated in Fig. \ref{fig:ExplainProb} it is easy to see that $\Pi_{s} = F_{g}\left( g_{\mathrm{min}}\right)$, since SemCom always achieves rates greater than that of the condition and it suffices to achieve the desired similarity threshold as described by the first branch of \eqref{eq:outProbSem}. By leveraging \eqref{eq:hybCombSemOut}, it is observed that in the case of $g_{\mathrm{bit}} \leq g_{\mathrm{min}}$, it holds that $\Pi_h = F_{g}\left( g_{\mathrm{bit}}\right)$, while for $g_{\mathrm{min}} \leq g_{\mathrm{bit}} \leq g_{\mathrm{max}}$, it follows that $\Pi_h = F_{g}\left( g_{\mathrm{min}}\right)$. Therefore, $\Pi_h = F_{g}\left( \min{\{g_{\mathrm{bit}},g_{\mathrm{min}} \}}\right)$ and in either case the outage probability of the hybrid network is smaller than or equal to that of pure BitCom and SemCom networks. These scenarios are visually depicted by $R_{\mathrm{out}} = R_1$ and $R_{\mathrm{out}} = R_2$ in Fig. \ref{fig:ExplainProb}. 
    
    \item If $A_{k,1} \leq kR_{\mathrm{out}} \leq A_{k,2}$ and $R_{\mathrm{out}} \leq R_{\mathrm{max}}$:
For $g \leq g_{\mathrm{max}}$, it holds that $R_{b}(g) \leq R_{s}(g)$.
    For the outage rate it is known that $R_{\mathrm{out}} = R_{b}(g_{\mathrm{bit}}) = R_{s}(g_{\mathrm{sem}})$. Without loss of generality, we assume that $g_{\mathrm{bit}} < g_{\mathrm{sem}}$. Then, considering the monotonicity of the functions, we have $R_{b}(g_{\mathrm{bit}}) \leq R_{s}(g_{\mathrm{bit}}) < R_{s}(g_{\mathrm{sem}})$, which is a contradiction since the ends of the inequality are equal. Therefore, $g_{\mathrm{sem}} \leq g_{\mathrm{bit}} \leq  g_{\mathrm{max}}$ must hold since $R_{\mathrm{out}} \leq R_{\mathrm{max}} \Rightarrow g_{\mathrm{bit}} \leq g_{\mathrm{max}}$. This result is also highlighted by the cases that correspond to $R_{\mathrm{out}} = R_3$ and $R_{\mathrm{out}} = R_4$ in Fig. \ref{fig:ExplainProb}, where $\tilde{g}_3 < \hat{g}_3 < g_{\mathrm{max}}$ and $\tilde{g}_4 < \hat{g}_4 < g_{\mathrm{max}}$ hold, respectively.

    Now, we can distinguish two possible scenarios regarding $g_{\mathrm{bit}}$.
    \begin{itemize}
        \item If $g_{\mathrm{bit}} \leq g_{\mathrm{min}}$:
        Since $g_{\mathrm{sem}} \leq g_{\mathrm{bit}} \leq g_{\mathrm{min}}$ must hold, for a hybrid user that utilizes SemCom, the condition of the second branch of \eqref{eq:singleSemOut} is satisfied. For the same reasoning, the outage probability of a pure SemCom user is given by the second branch of \eqref{eq:outProbSem} while for BitCom the outage probability is given by $F_{g}\left( g_{\mathrm{bit}}\right)$. Consequently, the third branch of \eqref{eq:hybCombSemOut} holds for the hybrid network, which results in $\Pi_h = F_{g}\left( g_{\mathrm{bit}}\right) = \Pi_{b}$ and $\Pi_h = F_{g}\left( g_{\mathrm{bit}}\right) \leq F_{g}\left( g_{\mathrm{min}}\right) = \Pi_{s}$. 

        \item If $g_{\mathrm{min}} \leq g_{\mathrm{bit}} \leq g_{\mathrm{max}}$: Since $g_{\mathrm{sem}} \leq g_{\mathrm{max}}$ must hold, for a hybrid user that utilizes SemCom, the condition of either the second or third branch of \eqref{eq:singleSemOut} is satisfied. Accordingly, the hybrid user outage probability is given by the fourth or fifth branch of \eqref{eq:hybCombSemOut}, while for pure SemCom it is given by the second or third branch of \eqref{eq:outProbSem}. By combining these, it can be observed that $\Pi_h$ and $\Pi_{s}$ are equal between them, since the discussed branches for each probability have the same condition with respect to $g_{\mathrm{sem}}$ and $g_{\mathrm{min}}$. On the other hand, for BitCom the outage probability is given by $F_{g}\left( g_{\mathrm{bit}} \right)$. Thus, by leveraging the above, we can conclude that        
        $\Pi_h  =  F_{g}\left( \max{\{g_{\mathrm{min}},g_{\mathrm{sem}} \}}\right) = \Pi_{s}$ and $\Pi_h \leq  F_{g}\left( g_{\mathrm{sem}}\right) \leq F_{g}\left( g_{\mathrm{bit}}\right) = \Pi_{b}$.
    \end{itemize}

    \begin{figure}
    \centering    \includegraphics[width=0.9\linewidth]{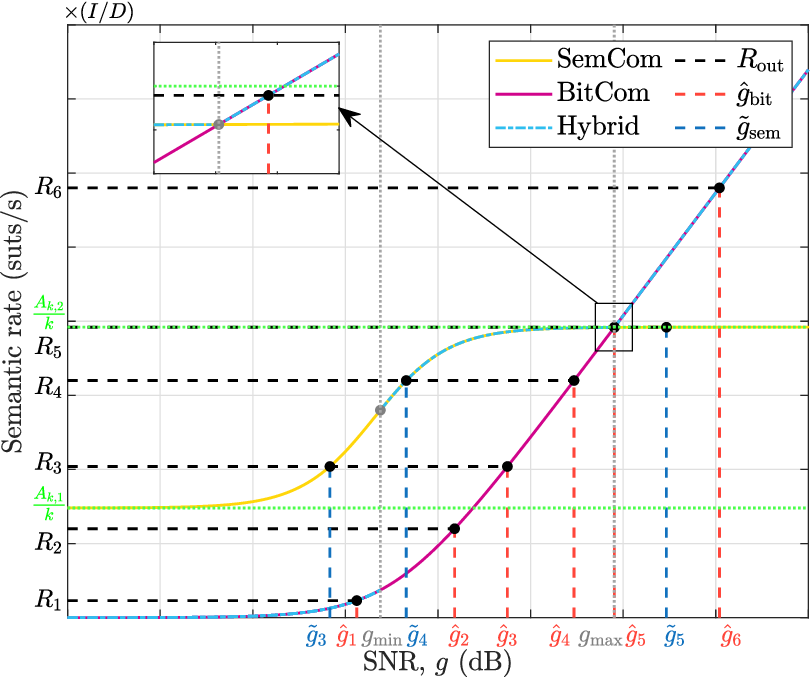}
    \vspace{-2mm}
    \caption{Semantic rate vs SNR. The black dashed lines correspond to outage rate thresholds, $R_{\mathrm{out}} \in \{R_1, \cdots, R_6 \}$. The red and blue dashed lines correspond to the same color SNR values of BitCom, $g_{\mathrm{bit}} \in \{\hat{g}_1, \cdots, \hat{g}_6 \}$, and SemCom, $g_{\mathrm{sem}} \in \{\tilde{g}_3, \cdots, \tilde{g}_5 \}$, respectively, that achieve the rate threshold, $R_{\mathrm{out}}$, of the same index, e.g., $R_{4} = R_{s}(\tilde{g}_4) = R_{b}(\hat{g}_4)$. The green dotted lines highlight the theoretical bounds of semantic rate for SemCom and the gray dotted lines define the interval of semantic utilization, $\mathcal{G}$.}
    \label{fig:ExplainProb}
\end{figure}

    \item If $A_{k,1} \leq kR_{\mathrm{out}} \leq A_{k,2}$ and $R_{\mathrm{max}} < R_{\mathrm{out}}$: For $g_{\mathrm{max}} < g$, it holds that $R_{s}(g) < R_{b}(g)$. Without loss of generality, we assume that $g_{\mathrm{sem}} \leq g_{\mathrm{bit}}$. Then, leveraging the fact that $R_{\mathrm{out}} = R_{b}(g_{\mathrm{bit}}) = R_{s}(g_{\mathrm{sem}})$ and the monotonicity of the functions, we have $R_{s}(g_{\mathrm{sem}}) \leq R_{s}(g_{\mathrm{bit}}) < R_{b}(g_{\mathrm{bit}})$, which is a contradiction since the ends of the inequality are equal. Thus, $g_{\mathrm{max}} < g_{\mathrm{bit}} <  g_{\mathrm{sem}}$ must hold since $R_{\mathrm{max}} < R_{\mathrm{out}} \Rightarrow g_{\mathrm{max}} < g_{\mathrm{bit}}$. This result is also shown by the scenario that corresponds to $R_{\mathrm{out}} = R_5$ in Fig. \ref{fig:ExplainProb}, where $g_{\mathrm{max}} < \hat{g}_5 < \tilde{g}_5$ holds.

    In this case, the outage probability of a hybrid user for BitCom and SemCom are given by the third and fourth branch of \eqref{eq:singleBitOut} and \eqref{eq:singleSemOut}, respectively. Consequently, it holds that $\Pi_h = F_{g}\left( g_{\mathrm{bit}}\right) =  \Pi_{b}$ and $\Pi_h = F_{g}\left( g_{\mathrm{bit}}\right) < F_{g}\left( g_{\mathrm{sem}}\right) = \Pi_{s},$
    since a pure SemCom user's  outage probability is given by the third branch of \eqref{eq:outProbSem}.

    \item If $A_{k,2} \leq kR_{\mathrm{out}}$: 
In this case it is straightforward to see that $\Pi_h = F_{g}\left( g_{\mathrm{bit}}\right) =  \Pi_{b}$. Then, since $\Pi_{s} = 1$, the hybrid network has outage probability smaller than or equal to that of pure BitCom and SemCom networks.
    \end{enumerate} 

Therefore, the hybrid network always achieves smaller outage probability compared to pure BitCom and SemCom networks.
\end{IEEEproof}

\subsection{Generalized Outage Probability}
Although network outage probability is a significant metric that can be useful for designing a network with enhanced reliability, modern communication systems are also interested in achieving good QoS. This metric can be quite insightful since network outage corresponds to extreme network situations that do not usually characterize a system.  Therefore, it is important to quantify the probability of any number of users being in outage. Additionally, driven by this metric, it is possible to design network characteristics in an optimal way according to the desired operational requirements.

For the proposed hybrid network, it is assumed that the number of users in outage, $L_o$, is such that $0 \leq L_{l} \leq L_o \leq L_u \leq L$, where $L_{l}$ and $L_u$ are the smallest and largest number of users that can potentially be in outage, respectively. Then, the probability that exactly $L_o$ users are in outage is given by
\begin{equation}\label{eq:circProb}
    p_{L,L_o}(R) = \binom{L}{L_{o}}  \Pi_{h}^{L_{o}}\left( 1 -  \Pi_{h} \right)^{L-L_{o}},
\end{equation}
where $\binom{\cdot}{\cdot}$ denotes the binomial coefficient.
Using \eqref{eq:circProb}, the probability that all users in $\{L_l, \cdots, L_u \}$ are in outage is expressed as
\begin{align}\label{eq:totProb}
    \mathcal{S}_{L_l,L_u}(R) \!=\! \!
    \sum_{m=L_l}^{L_u} \!\! p_{L,m}(R)
    \!= \!\!\sum_{m=L_l}^{L_u} \!\! \binom{L}{m}\Pi_{h}^{m}\left( 1 -  \Pi_{h} \right)^{L-m}.
\end{align}

Although it is important to calculate the probability of any number of users being in outage, it is particularly interesting to find design parameters that achieve certain probability thresholds for a minimum QoS. Specifically, a network can determine the cell radius to ensure some desired criteria. In this direction, we are interested in calculating the cell radius, $R$, such that $L_l$ users or more will have small probability of being in outage. In this case, $L_u = L$ and using the binomial distribution CDF and \eqref{eq:totProb}, we have
\begin{equation}\label{eq:IrBetaPth}
    \mathcal{S}_{L_l,L}(R) = 1 - \mathcal{S}_{0,L_l-1}(R) = I_{\Pi_{h}}(L_l,L-L_l+1),
\end{equation}
where $I_{\Pi_{h}}(\cdot,\cdot)$ denotes the regularized beta incomplete function at the integral upper limit $\Pi_{h}$ and the fundamental  property $I_{1-\Pi_{h}}\!(L\!-\!L_l\!+\!1,L_l) \!=\! 1\!-\!I_{\Pi_{h}}\!(L_l,L\!-\!L_l\!+\!1)$ has been used. Therefore, for the network outage probability to be up to a desired probability threshold, $P_{\mathrm{th}}$, the following inequality must hold:
\begin{equation}\label{eq:PthIneq}
    I_{\Pi_{h}}(L_l,L-L_l+1) \leq P_{\mathrm{th}}.
\end{equation}
Based on \eqref{eq:PthIneq}, the cell radius must be upper bounded and consequently, it suffices to calculate the maximum radius that satisfies the inequality, which is achieved for the equality case. By taking advantage of the inverse regularized beta incomplete function, the equality of \eqref{eq:PthIneq}can by definition be equivalently expressed as
$\Pi_h = I_{P_{\mathrm{th}}}^{-1}(L_l,L-L_l+1)$. Then, by leveraging \eqref{eq:hybCombSemOut}, the radius for which $P_{\mathrm{th}}$ is achieved can be found by solving
\begin{equation}\label{eq:BetaEq}
    {}_1\!F_1\left(\!\frac{2}{a};1\!+\!\frac{2}{a};-\frac{y_{\mathrm{th}}}{c_L}R^a\right) = u_{\mathrm{th}},
\end{equation}
where $u_{\mathrm{th}} = 1 - I_{P_\mathrm{th}}\!(L_l,L\!-\!L_l\!+\!1) $ is set for convenience and $y_{\mathrm{th}} \in \{ g_{\mathrm{bit}}, g_{\mathrm{min}}, g_{\mathrm{sem}} \}$ is such that $y_{\mathrm{th}}$ denotes the corresponding SNR argument of the CDF as shown in \eqref{eq:hybCombSemOut}. It should be noted that the same way can be used to derive the radius for which the outage probability is below a desired threshold for a BitCom network, where $y_{\mathrm{th}}$ is substituted by $g_{\mathrm{bit}}$, regardless of the operating network parameter values.  
In general, \eqref{eq:BetaEq} does not have a closed-form solution in terms of $R$. Nevertheless, for the case of $a=2$, a solution of \eqref{eq:BetaEq} can be derived as given by the following proposition.

\begin{proposition}\label{lemma:optRadPth}
If there are free space path loss conditions, the cell radius, $R_{\mathrm{th}}^*$, for which $L_l$ or more users are in outage with probability $P_{\mathrm{th}}$, is given by 
\begin{equation}\label{eq:radius_Pth}
    R_{\mathrm{th}}^{*}\! = \! \sqrt{\frac{c_L}{y_{\mathrm{th}}}\left( \frac{1}{u_{\mathrm{th}}} + W_0 \left(-\frac{1}{u_{\mathrm{th}}}\exp\left( -\frac{1}{u_{\mathrm{th}}}\right) \right)\right)},
\end{equation}
where $W_0(\cdot)$ denotes the principal branch of the Lambert W function.
\end{proposition}
\begin{IEEEproof}
The proof is provided in Appendix \ref{app:proof_lemPth}.
\end{IEEEproof}

Therefore, depending on the QoS that must be satisfied, the cell radius, $R$, of a hybrid network can be selected to satisfy the condition $R \leq R_{\mathrm{th}}^{*}$ in order to ensure that the generalized outage probability is below the desired levels.

\subsection{Semantic Utilization}
As highlighted by the previous analysis, the studied hybrid network can achieve considerably increased performance, both in terms of total semantic rate and outage probability, due to the integration of SemCom. Nevertheless, incorporating the latter into networks can pose several challenges in practice. In more detail, the utilization of SemCom can be computationally challenging for a system since the use of the DeepSC encoder and decoder leads to increased requirements in terms of resources for the weights of the DNN to be stored for usage. Based on this, it is possible that a system can handle up to a certain number of users that can simultaneously utilize semantic transmission. 

On the other hand, semantic transmission can lead to decreased delay and outage, which makes it preferable to use whenever possible. Additionally, the cost for the deployment of specialized hardware to support SemCom should also be considered in the network design. For instance, if only few users would statistically utilize SemCom, the dedicated resources are not efficiently managed since they would not be fully exploited for transmission, and consequently their cost would not result in a notable performance boost. Therefore, to ensure that SemCom is utilized by the system in order to minimize the transmission delay, while also leveraging memory and computational resources as best as possible, a practical system should aim to support semantic utilization within desired levels to take full advantage of its available resources. 

In this direction, semantic utilization is a metric that can capture the number of users that are served by SemCom transmission. With this in mind, the network can be designed in a way that would ensure the efficient and effective management of the dedicated resources by allowing a desired range of users to utilize SemCom. The probability that $L_s$ users are above the required outage rate, $R_{\mathrm{out}}$, and  use SemCom is given by
\begin{equation}\label{eq:semUtil}
    f_{L,L_s}(R) = \binom{L}{L_{s}}  \Pi_{g}^{L_{s}}\left( 1 -  \Pi_{g} \right)^{L-L_{s}},
\end{equation}
where $\Pi_{g} = \mathbb{P}\left(g \in \{ \mathcal{A}_s^c \cap \mathcal{G} \} \right)$ denotes the probability that a user prefers SemCom to BitCom transmission. By leveraging \eqref{eq:CDF}, $\Pi_g$ is derived in closed form by \eqref{eq:probUtil} at the top of the next page.

\begin{figure*}[!t]
\begin{equation}\label{eq:probUtil}
\Pi_{g} = \begin{cases} {}_1\!F_1\!\!\left(\frac{2}{a};1\!+\!\frac{2}{a};\!-\frac{g_{\mathrm{min}}}{c_L}R^a\!\right) \!-\! {}_1\!F_1\!\!\left(\frac{2}{a};1\!+\!\frac{2}{a};\!-\frac{g_{\mathrm{max}}}{c_L}R^a\!\right), & kR_{\mathrm{out}} \leq A_{k,1}, \\
{}_1\!F_1\!\!\left(\frac{2}{a};1\!+\!\frac{2}{a};\!-\frac{g_{\mathrm{min}}}{c_L}R^a\!\right) \!-\! {}_1\!F_1\!\!\left(\frac{2}{a};1\!+\!\frac{2}{a};\!-\frac{g_{\mathrm{max}}}{c_L}R^a\!\right), & g_{\mathrm{sem}} \leq g_{\mathrm{min}} \leq g_{\mathrm{max}}, A_{k,1} \leq kR_{\mathrm{out}} \leq A_{k,2}, \\
{}_1\!F_1\!\!\left(\frac{2}{a};1\!+\!\frac{2}{a};\!-\frac{g_{\mathrm{sem}}}{c_L}R^a\!\right) \!-\! {}_1\!F_1\!\!\left(\frac{2}{a};1\!+\!\frac{2}{a};\!-\frac{g_{\mathrm{max}}}{c_L}R^a\!\right), & g_{\mathrm{min}} \leq g_{\mathrm{sem}} \leq g_{\mathrm{max}}, A_{k,1} \leq kR_{\mathrm{out}} \leq A_{k,2}, \\
0, & g_{\mathrm{min}} \leq g_{\mathrm{max}} \leq g_{\mathrm{sem}}, A_{k,1} \leq kR_{\mathrm{out}} \leq A_{k,2}, \\
0, & A_{k,2} \leq kR_{\mathrm{out}}
\end{cases}
\end{equation}
\hrulefill
\end{figure*}

Consequently, the optimal cell radius for $L_l$ to $L_u$ users to utilize SemCom can be determined by solving $\sum_{m=L_l}^{L_u} df_{L,m}/dR = 0$ or by exploiting the monotonicity of the function within a certain interval. In more detail, taking the derivative of \eqref{eq:circProb} and performing some algebraic manipulations, the optimality conditions are derived by the following proposition.

\begin{proposition}\label{lemma:optRad}
The optimal cell radius for all users in $\{L_l, \cdots, L_u\}$ to utilize SemCom is determined by solving the following optimality condition equations:
\begin{equation}\label{eq:opt_cond}
        \Pi_{g} = \frac{1}{ \!\!1\!+\!\!\left(\!\frac{\binom{L-1}{L_u}}{\binom{L-1}{L_l-1}}\!\right)^{\!\!\frac{1}{L_u-L_l+1}}} \text{ and } \frac{d\Pi_{g}}{dR} = 0,
\end{equation}
where $\frac{d\Pi_{g}}{dR}$ is given by \eqref{eq:probUtilDer} at the top of the next page.
\end{proposition}
\begin{IEEEproof}
The proof is provided in Appendix \ref{app:proof_lemDer}.
\end{IEEEproof}

Therefore, depending on the network parameters, the resulting semantic utilization of a single user, $\Pi_g$, can be considered as a network curve through which the optimal cell radius to achieve semantic utilization within desired levels can be calculated by solving \eqref{eq:opt_cond}.

\begin{figure*}[t]
\vspace{-1mm}
\begin{equation}\label{eq:probUtilDer}
\frac{d\Pi_{g}}{dR} \!= \!\begin{cases} \! \frac{1}{R} \!\!\left[ \!{}_1\!F_1\!\!\left(\frac{2}{a};\!1\!+\!\frac{2}{a};\!-\frac{g_{\mathrm{max}}}{c_L}R^a\!\right) \!-\!{}_1\!F_1\!\!\left(\frac{2}{a};\!1\!+\!\frac{2}{a};\!-\frac{g_{\mathrm{min}}}{c_L}R^a\!\right) \!-\! \frac{1}{e^{\frac{g_{\mathrm{max}}}{c_L}R^a}} \!+\! \frac{1}{e^{\frac{g_{\mathrm{min}}}{c_L}R^a}} \! \right] \!\!, & \!\!\!\! kR_{\mathrm{out}} \leq A_{k,1}, \\
\! \frac{1}{R} \!\!\left[ \!{}_1\!F_1\!\!\left(\frac{2}{a};\!1\!+\!\frac{2}{a};\!-\frac{g_{\mathrm{max}}}{c_L}R^a\!\right) \!-\! {}_1\!F_1\!\!\left(\frac{2}{a};\!1\!+\!\frac{2}{a};\!-\frac{g_{\mathrm{min}}}{c_L}R^a\!\right) \!-\! \frac{1}{e^{\frac{g_{\mathrm{max}}}{c_L}R^a}} \!+\! \frac{1}{e^{\frac{g_{\mathrm{min}}}{c_L}R^a}} \! \right] \!\!, & \!\!\!\! g_{\mathrm{sem}} \! \leq \!  g_{\mathrm{min}} \! \leq \! g_{\mathrm{max}}, A_{k,1} \! \leq \! kR_{\mathrm{out}} \! \leq \! A_{k,2}, \\
\! \frac{1}{R} \!\!\left[ \!{}_1\!F_1\!\!\left(\frac{2}{a};\!1\!+\!\frac{2}{a};\!-\frac{g_{\mathrm{max}}}{c_L}R^a\!\right) \!-\! {}_1\!F_1\!\!\left(\frac{2}{a};\!1\!+\!\frac{2}{a};\!-\frac{g_{\mathrm{sem}}}{c_L}R^a\!\right) \!-\! \frac{1}{e^{\frac{g_{\mathrm{max}}}{c_L}R^a}} \!+\! \frac{1}{e^{\frac{g_{\mathrm{sem}}}{c_L}R^a}} \! \right]\!\!, & \!\!\!\! g_{\mathrm{min}} \! \leq \! g_{\mathrm{sem}} \! \leq \! g_{\mathrm{max}}, A_{k,1} \! \leq \! kR_{\mathrm{out}} \! \leq \! A_{k,2}, \\
0, & \!\!\!\! g_{\mathrm{min}} \! \leq \! g_{\mathrm{max}} \! \leq \! g_{\mathrm{sem}}, A_{k,1} \! \leq \! kR_{\mathrm{out}} \! \leq \! A_{k,2}, \\
0, &\!\!\!\!  A_{k,2} \leq kR_{\mathrm{out}}
\end{cases}
\end{equation}
\hrulefill
\end{figure*}

\section{Simulation Results And Discussion}\label{sec:numerical}
In this section, the simulation results along with the derived theoretical results are presented. For the results, Monte Carlo simulations have been performed over $10^8$ channel and user location realizations. Unless otherwise stated, the values of the parameters are given in Table \ref{table:tableParameters}. Moreover, regarding DeepSC, the curve-fitting constants of semantic similarity are set equal to $A_{k,1} = 0.37, A_{k,2} = 0.98, C_{k,1} = 0.2525, C_{k,2} = -0.7895$ as given in \cite{SemNOMA}. In the following figures, the theoretical results are illustrated by colored lines, while the numerical results are illustrated by marks of the same color.

\begin{table}[t] 
    \centering
    \caption{Simulation Parameters}
    \label{table:tableParameters}
    \begin{tabular}{c|c}
    \textbf{Parameter} & \textbf{Value (Unit)} \\
    \hline \hline
    Noise power spectral density, $N_0$ & $-174$ dBm/Hz \\
    \hline
    User bandwidth, $W$ & $20$ MHz \\
    \hline
    Frequency, $f_c$ & $2.4$ GHz \\
    \hline
    Transmit Power per user, $\mathcal{P}$ & $1$ W \\
    \hline
    Path loss exponent, $a$ & $3$ \\
    \hline
    Semantic symbols, $k$ & $5$ symbols/word \\
    \hline
    Bit symbols, $\mu$ & $40$ symbols/word \\
    \hline
    Similarity threshold, $M^{\mathrm{th}}$ & $0.75$ \\
    \hline
    Number of users, $L$ & $30$ \\
    \hline
    Outage rate threshold, $R_{\mathrm{out}}$ & $0.04$ suts/s\\
    \hline
    BitCom BER, $\mathrm{BER}$ & $10^{-3}$ \\
    \end{tabular} 
\end{table}

\begin{figure}
    \centering
    \includegraphics[width=0.9\linewidth]{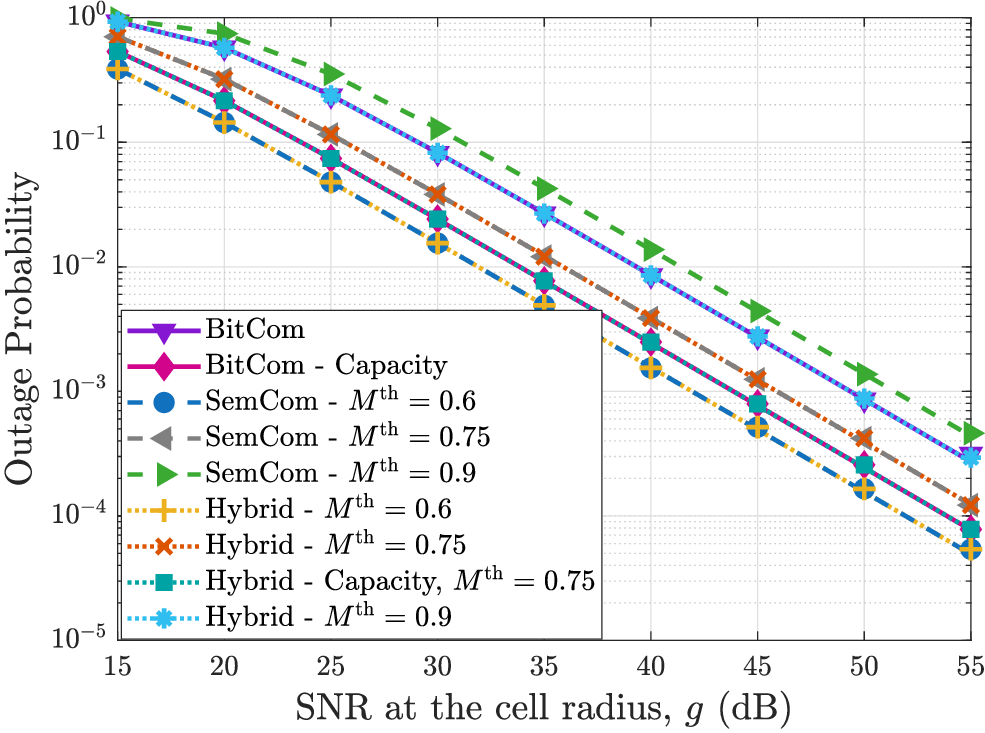}
    \caption{Probability of at least one user being in outage for all types of networks and varying similarity threshold $M^{\mathrm{th}}$. For the BitCom and hybrid network, the $\mathrm{BER}$ value of the bit users is set to that of Table \ref{table:tableParameters} except for the capacity curves, where $\Gamma = 1$ and no error occurs for bit transmission.}
    \label{fig:SNR_SimConstraints}
\end{figure}

In Fig. \ref{fig:SNR_SimConstraints}, the network outage probability of at least one user being in outage is shown. As can be observed, the increasing value of the similarity threshold, $M^{\mathrm{th}}$, does not affect the performance of BitCom as $g_{\mathrm{bit}}$ remains unchanged. On the other hand, the performance of SemCom worsens as $g_{\mathrm{min}}$ increases. While SemCom outperforms BitCom for small and mid-range values of $M^{\mathrm{th}}$, the opposite behavior occurs for high values of the latter. This behavior can be explained by considering the proof of Theorem \ref{th:NetworkPerf}. In more detail, as explained in Theorem \ref{th:NetworkPerf}, when $R_{\mathrm{out}} < R_{\mathrm{max}}$, if BitCom can be used to serve a user, SemCom can also serve the same user, except in scenarios that $g_{\mathrm{bit}} \leq g_{\mathrm{min}}$, where BitCom is preferable. As a result, the hybrid network utilizes the same transmission method for all users, resulting in performance similar to the best of the other two networks. It is also significant to highlight the effect of the factor $\Gamma$ in the BitCom network. Specifically, for the case that the capacity is achieved, $g_{\mathrm{bit}}$ is considerably smaller than its corresponding value when $\mathrm{BER}=10^{-3}$ is considered for the uncoded scheme, which results in improved outage probability. In fact, in this case, BitCom outperforms SemCom even for the mid-range similarity threshold of $M^{\mathrm{th}}=0.75$. Therefore, the achievable rate of BitCom transmission is expected to affect the performance of both BitCom and hybrid networks.   

\begin{figure}
    \centering
    \includegraphics[width=0.9\linewidth]{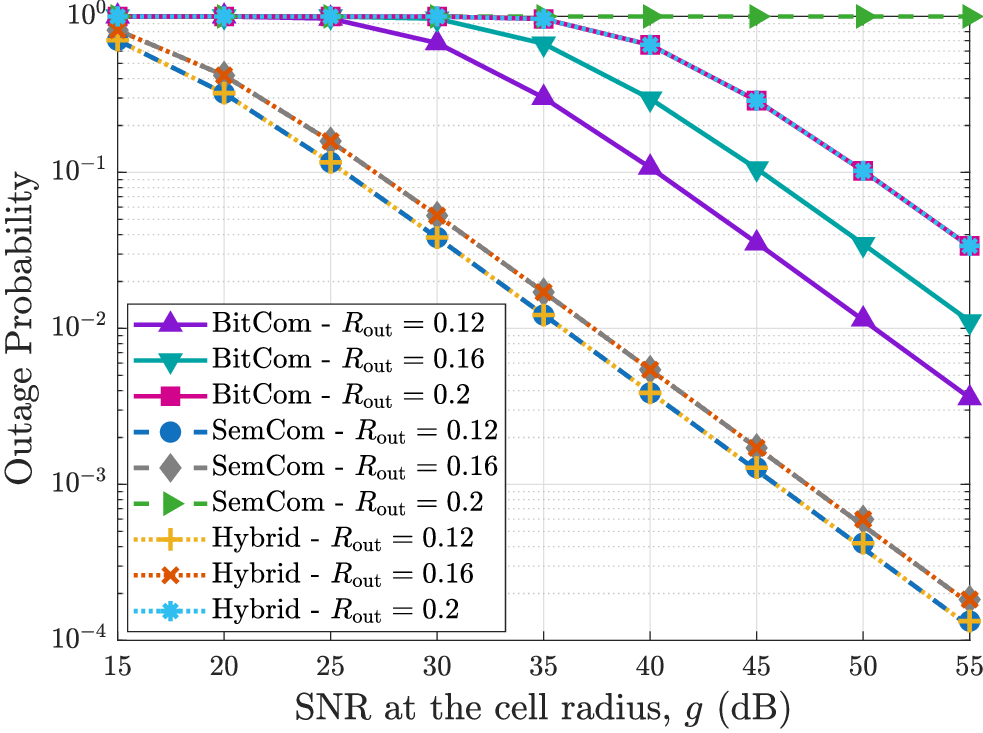}
    \caption{Probability of at least one user being in outage for all types of networks and varying outage rate threshold $R_{\mathrm{out}}$.}
    \label{fig:SNR_Ra}
\end{figure}

In Fig. \ref{fig:SNR_Ra}, the outage probability is plotted versus the received SNR at the cell radius for varying outage rate thresholds and all types of networks. Similar to Fig. \ref{fig:SNR_SimConstraints}, it is observed that the hybrid network achieves the performance of the best of the other two types of networks for each case. As illustrated, the performance of BitCom for varying $R_{\mathrm{out}}$ changes since $g_{\mathrm{bit}}$ also varies. Furthermore, it is important to note that up to mid-range values of $R_{\mathrm{out}}$, SemCom outperforms BitCom by more than one order of magnitude. However, SemCom cannot be selected for high outage rate threshold values because its rate saturates below the desired thresholds, as also illustrated in Fig. \ref{fig:ExplainProb}, resulting in continuous outage. Therefore, it is evident that the hybrid network is crucial for achieving the best possible performance and avoiding complete network outage in extreme cases. 

\begin{figure}
    \centering
    \includegraphics[width=0.9\linewidth]{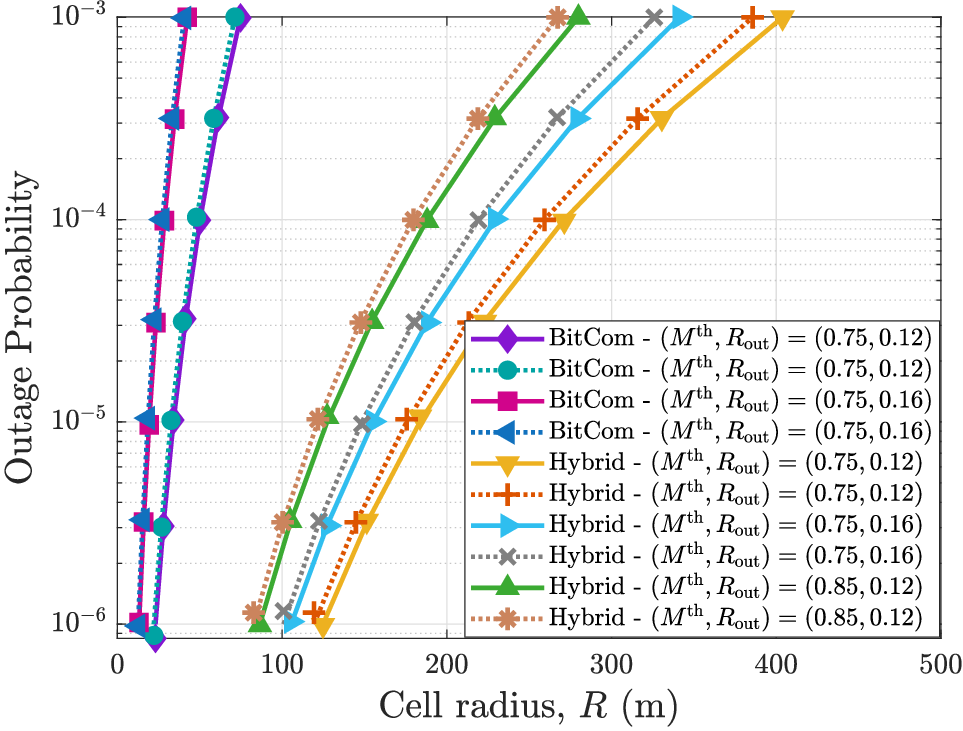}
    \caption{Probability of $L_l=3$ or more users being in outage for increasing cell size and varying similarity and outage rate thresholds, when $a=2$, $\mathcal{P} = 1mW$. The continuous and dotted lines symbolize $L=10$ and $L=50$ users, respectively.}
    \label{fig:R_Ra_SimConstraints}
\end{figure}

In Fig. \ref{fig:R_Ra_SimConstraints}, the probability of $3$ or more users being in outage is presented for a varying number of users in the cell and network parameters for the BitCom and hybrid network. As can be seen, the simulation results obtained by Lemma \ref{lemma:optRadPth} and \eqref{eq:radius_Pth} accurately calculate the cell radius that should be used to achieve a desired outage rate threshold. It is also important to highlight how the number of users and network parameter values affect the calculated cell radius. Specifically, as $L$ increases, the cell radius decreases, which can be expected, as a larger radius would result in more users being in outage. Moreover, increasing the outage rate threshold from $R_{\mathrm{out}}=0.12$ to $R_{\mathrm{out}}=0.16$ leads to a smaller cell radius and a  performance decrease of approximately half an order of magnitude for the hybrid network and one order of magnitude for BitCom, respectively. In addition, increasing the similarity threshold results in an even smaller cell radius, which is attributed to the fact that a large value of $M^{\mathrm{th}}$ leads to a greater value of $g_{\mathrm{sem}}$ compared to $g_{\mathrm{min}}$ for a specific outage rate threshold. Thus, by taking advantage of \eqref{eq:SNRmin} and \eqref{eq:SNRoutSem}, we can compare $M^{\mathrm{th}}$ and $kR_{\mathrm{out}}$ to estimate which of the two network parameters determines the cell radius. Consequently, the theoretical analysis of Section \ref{sec:perfAnalysis} allows us to correctly design the cell radius according to the desired network parameters.

\begin{figure}
    \centering
    \includegraphics[width=0.9\linewidth]{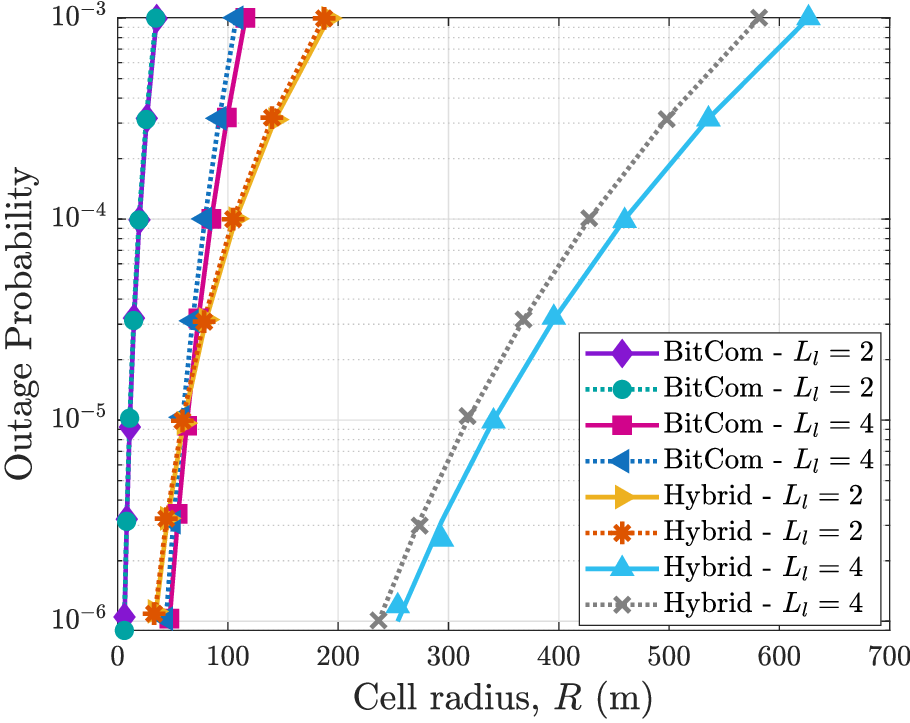}
    \caption{Probability of $L_l = 2, L_l = 4$ or more users being in outage for increasing cell size and outage rate threshold $R_{\mathrm{out}}=0.12$, when $a=2$, $\mathcal{P} = 1mW$. The continuous and dotted lines symbolize $L=10$ and $L=50$ users, respectively.}
    \label{fig:R_Ll}
\end{figure}

Fig. \ref{fig:R_Ll} illustrates the generalized outage probability studied for two scenarios, $L_l = 2$ and $L_l = 4$. It is important to note that, although the number of users in the cell, $L$, affects the radius that should be used, the difference in radius for varying values of $L$ is only a few meters. However, the cell radius is primarily affected by the smallest desired number of users that can be in outage, i.e., $L_l$. In more detail, the smaller the value of $L_l$, the smaller the cell radius should be to ensure the desired outage probability levels. In addition, the studied hybrid network outperforms the conventional BitCom network, thus utilizing it can increase spectral efficiency because it can support larger cells than current networks. This characteristic is most pronounced for larger values of $L_l$, where the cell radius that the hybrid network can select is greater than that of BitCom by more than $200$m for stricter outage requirements ($10^{-6}$) and more than $500$m for relatively less strict outage requirements ($10^{-3}$). Nevertheless, even with a more reliable network design characterized by smaller $L_l$ values, the hybrid network can use a cell radius three times larger than conventional BitCom.

\begin{figure}
    \centering
    \includegraphics[width=0.9\linewidth]{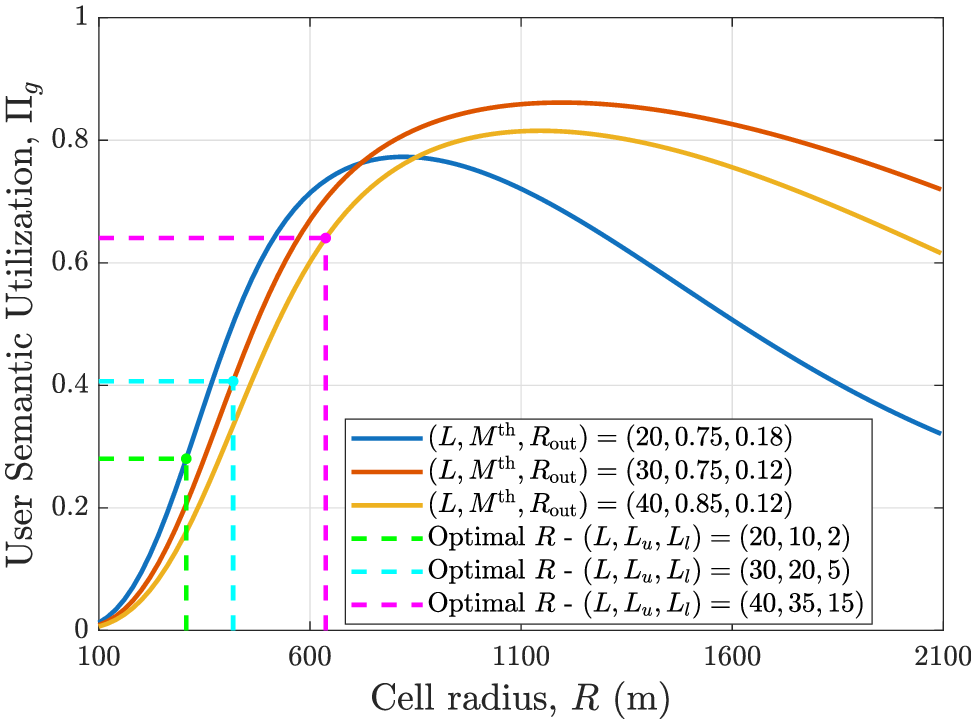}
    \caption{Semantic utilization probability of a user in a cell of $L$ users. for varying similarity and outage rate threshold.}
    \label{fig:R_Pg}
\end{figure}

In Fig. \ref{fig:R_Pg}, the probability of semantic utilization for a user in a cell that contains $L$ users is illustrated. Using Lemma \ref{lemma:optRad}, we calculate the optimal cell radius that achieves semantic utilization within a desired range $\{L_l , \cdots, L_u \}$. As can be observed, each network curve $\Pi_g$ exhibits the same behavior, increasing up to a maximum and then decreasing. It is important to emphasize that the equations given by \eqref{eq:opt_cond} can be used to calculate the cell radius for any combination of network parameters and $L_s, L_l, L_u$. It is evident that if the value described by the right hand-side of the equation containing $\Pi_g$ in \eqref{eq:opt_cond} cannot be obtained, then only the derivative equation can yield the minimum of the semantic utilization. It should also be noted that the first equation of \eqref{eq:opt_cond} can have two distinct roots, but the second one is not depicted here, since it corresponds to a larger cell size that would lead to smaller achievable rates within the cell. Nevertheless, the system could select the second root for the radius, if semantic utilization was the primary concern of the system, which would lead to higher spectral efficiency, at the expense of semantic rate for all users.

\begin{figure}
    \centering
    \includegraphics[width=0.9\linewidth]{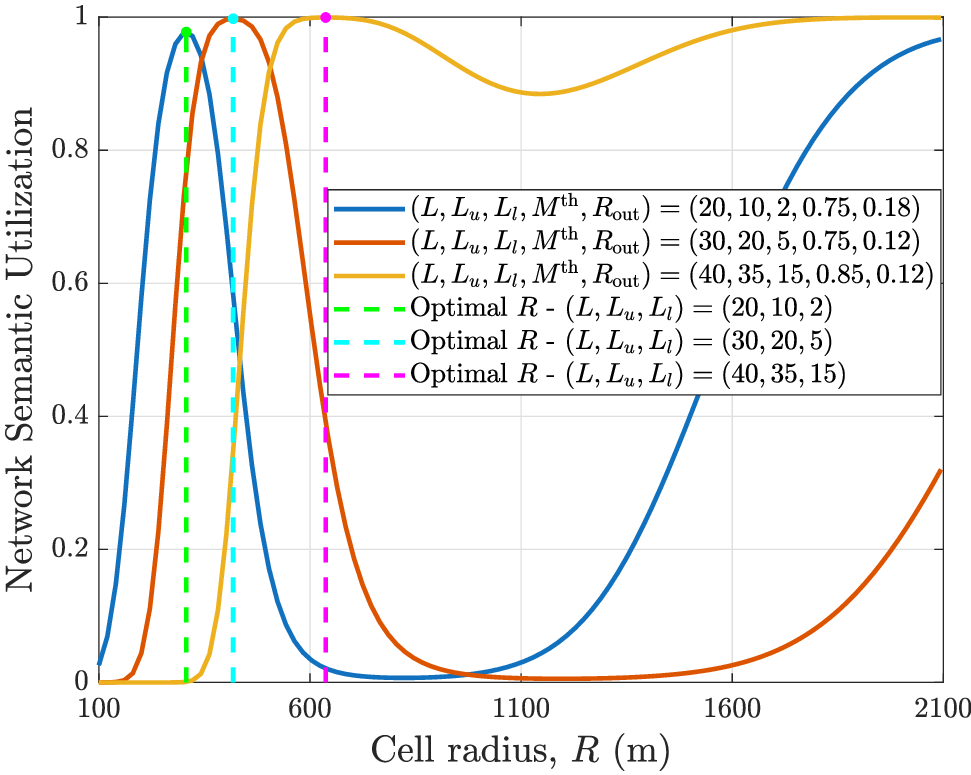}
    \caption{Semantic utilization probability for $L_l$ to $L_u$ users in a cell of $L$ users for varying similarity and outage rate threshold.}
    \label{fig:R_SemUtil}
\end{figure}

Fig. \ref{fig:R_SemUtil} presents the semantic utilization for $L_l$ to $L_u$ users in a cell of $L$ users. As highlighted, the radius calculated by equation \eqref{eq:opt_cond} achieves the maximum semantic utilization for the desired range of users. Moreover, all curves follow the expected behavior, exhibiting their minimum at the corresponding maximum of $\Pi_g$ as indicated in Fig. \ref{fig:R_Pg}, and increasing up to the second point of maximization which corresponds to the second root of the first equation given in \eqref{eq:opt_cond}. It should be noted that the semantic utilization showcases a smoother behavior near its maximum as $L$ increases. Although this can be affected by the selection of $L_l, L_u$, it is expected that, in a cell containing a large number of users, a small change in radius may cause a user to not utilize SemCom transmission, however the overall utilization would remain almost unchanged when a wide range $\{L_l, \cdots, L_u \}$ is desired. Therefore, even in a resource-constrained system that cannot serve all users through SemCom, an optimal radius can be found that allows the system to use SemCom within a desired range. This ensures that the system takes full advantage of its computational and memory resources.


\section{Conclusions}\label{sec:conc}
In this work, the performance of a hybrid BitCom-SemCom network was considered in terms of outage probability. To enable a more flexible network structure that can tolerate a certain number of users being in outage, a generalized outage probability was introduced and examined. Furthermore, the cell radius that must be used to achieve a desired outage rate threshold while simultaneously satisfying a common similarity threshold was determined,  and a closed-form expression was derived for a specific path loss exponent. Moreover, considering the substantial memory and computational resources required by a network to serve a large number of users as well as the cost of the required specialized hardware that must be deployed, a semantic utilization metric was proposed. Leveraging this metric allows the network to select a cell radius that maximizes the probability of using SemCom for a desired range of users, which in practice is dictated by the amount of available dedicated resources. Thus, the resources can be fully exploited to boost the performance of the hybrid network. Simulation results validated the theoretical analysis and provided important insights for designing a hybrid cellular network tailored to specific requirements. Based on these insights, a reliable and spectral-efficient hybrid network can be designed by selecting a cell radius that fully satisfies the network's requirements or achieves an acceptable trade-off between reliability and resource management.

\appendices
\section{Proof of Proposition \ref{lemma:optRadPth}}\label{app:proof_lemPth}
By setting $z=-(y_{\mathrm{th}}R^2) / c_L$ and leveraging the fact that ${}_1\!F_{1}\left(1; 2; z \right) = (e^z - 1) / z$, \eqref{eq:BetaEq} can be rewritten as 
\begin{equation}\label{eq:a2hyp}
 e^z - 1 = u_{\mathrm{th}}z ,
\end{equation}
where the first equation is equivalent to the second one when $z \neq 0$. It can be observed that \eqref{eq:a2hyp} has only one feasible solution, since the line described by the right-hand side of the equation intercepts the exponential curve at exactly two points, one of which is $z=0$ and thus is discarded. Then, by performing algebraic manipulations, we can equivalently write \eqref{eq:a2hyp} as
\begin{align}\label{eq:a2W}
    e^{-z} = \frac{1}{u_{\mathrm{th}}z+1} &\Rightarrow \left(u_{\mathrm{th}}z+1\right)e^{\left(-z-\frac{1}{u_{\mathrm{th}}}\right)} = e^{-\frac{1}{u_{\mathrm{th}}}} \nonumber \\ 
    &\!\!\!\!\!\!\!\!\Rightarrow -\!\left(\!z\!+\!\frac{1}{u_{\mathrm{th}}}\!\right)e^{-\left(z+\frac{1}{u_{\mathrm{th}}}\right)} \!=\! -\frac{e^{-\frac{1}{u_{\mathrm{th}}}}}{u_{\mathrm{th}}} \nonumber \\
    &\!\!\!\!\!\!\!\!\Rightarrow -z-\!\frac{1}{u_{\mathrm{th}}} \!=\! W_0 \!\left(\! -\frac{1}{u_{\mathrm{th}}} \exp \!\left(\! -\frac{1}{u_{\mathrm{th}}}\right) \!\right)\!.
\end{align}
Substituting $z$ into the last equation, the cell radius, $R_{\mathrm{th}}^*$, that satisfies \eqref{eq:a2W} can be found to be given by \eqref{eq:radius_Pth}, which completes the proof. 

\section{Proof of Proposition \ref{lemma:optRad}}\label{app:proof_lemDer}
By leveraging the product derivative rule and performing algebraic manipulations on binomials, the following recursive formula for $df_{L,m}/dR$ can be derived 
\begin{align}\label{eq:der_rec}
    \frac{df_{L,m}}{dR} &= \frac{d\Pi_{g}}{dR} L\binom{L-1}{m-1}\Pi_{g}^{m-1}(1-\Pi_{g})^{L-m} \nonumber \\
    &\quad- \frac{d\Pi_{g}}{dR}L\binom{L-1}{m}\Pi_{g}^{m}(1-\Pi_{g})^{L-m-1} \nonumber \\
    &= \frac{d\Pi_{g}}{dR}L\left( f_{L-1,m-1}(R) - f_{L-1,m}(R)\right),
\end{align}
when $1 \leq m \leq L-1$ must hold for the exponents in \eqref{eq:circProb} to be non-zero. For the extreme cases, $m = 0$ and $m = L$, similar expressions can be found and the formula is reduced to 
\begin{align}\label{eq:extCase}
    \frac{df_{L,m}}{dR}=\begin{cases}
    -\frac{d\Pi_{g}}{dR}Lf_{L-1,0}(R), &\quad m=0\\
    \frac{d\Pi_{g}}{dR}Lf_{L-1,L-1}(R), &\quad m=L .
    \end{cases}
\end{align}
Taking into account both \eqref{eq:der_rec} and \eqref{eq:extCase}, we have
\begin{align}\label{eq:telescope_sum1}
    \sum_{m=L_l}^{L_u} & \frac{df_{L,m}}{dR} \nonumber \\
    &\!\!\!\!=\begin{cases}
    \!\frac{d\Pi_{g}}{dR}L\!\left(f_{\!L\!-\!1,L_l-1}(R) \!-\! f_{\!L\!-\!1,L_u}(R) \right)\!,& \!\!\!\! L > L_u ,  L_l  > 0,\\
    \!\frac{d\Pi_{g}}{dR}Lf_{L-1,L_l-1}(R),& \!\!\!\!L_u=L, L_l>0, \\
    \!-\frac{d\Pi_{g}}{dR}Lf_{L-1,L_u}(R),& \!\!\!\!L > L_u, L_l=0,
    \end{cases}
\end{align}
since terms cancel out in the telescopic sum. Clearly, for the cases of $L_l=0$ and $L_u=L$, the extreme values can be achieved only when $d\Pi_{g}/dR = 0$, the solutions of which can be numerically found. Note that $d\Pi_{g}/dR$ can be explicitly derived by leveraging that $d\gamma\left( \frac{2}{a},\frac{y}{c_L}R^a\right) / dR = aR^{a-1}\left( \frac{y}{c_L}R^a \right)^{\frac{2}{a}-1} e^{-\frac{y}{c_L}R^a}$ along with fundamental properties of $\gamma(\cdot, \cdot)$, which results in \eqref{eq:probUtilDer}. 
For the general case, described by the first branch of \eqref{eq:telescope_sum1}, the condition $f_{L-1,L_l-1}(R) = f_{L-1,L_u}(R)$ can also lead to local extrema. By dividing both sides with $\Pi_{g}^{L_u}(1\!-\!\Pi_{g})^{L-L_u-1}$, the aforementioned equation can be reduced to the first equality in \eqref{eq:opt_cond}, which completes the proof.

\bibliographystyle{IEEEtran}
\bibliography{MyRefs}

\end{document}